\documentclass[reqno,11pt]{article}%
\usepackage{amssymb,amsmath,amsthm,setspace}
\usepackage[top=2cm, bottom=2cm, left=2cm, right=2cm]{geometry}
\usepackage{titling}

\newcommand{\subtitle}[1]{%
  \posttitle{%
    \par\end{center}
    \begin{center}{\it #1}\end{center}
    \vskip0.5em}%
}
\newtheorem{theorem}{Theorem}

\newtheorem{lemma}[theorem]{Lemma}

\numberwithin{equation}{section}
\newcommand{\R}{{\mathbb R}}         
\newcommand{\Z}{{\mathbb Z}}

\renewcommand{\a}{\alpha}

\newcommand{\g}{\gamma}
           \newcommand\om{\omega}

\newcommand{\pa}{\partial}

\newcommand{\D}{\Delta}

\let \e=\varepsilon
\let\pa=\partial

\def\a{\alpha}
\def\b{\beta}
\def\x{\xi}

\def\be{\begin{equation}}
\def\ee{\end{equation}}
\def\bea{\begin{eqnarray}}
\def\eea{\end{eqnarray}}

\def\nn{\nonumber}

\def\b{\beta}
\numberwithin{theorem}{section}
\begin{document}\title{Macroscopic description of microscopically strongly
inhomogenous systems}
\subtitle{\small \it A mathematical basis for the synthesis of higher gradients metamaterials}
\author{{A. Carcaterra$^{1,2,3}$, F. dell'Isola$^{4,3}$, R. Esposito$^3$ and M. Pulvirenti$^{5,3}$}}
\date{}
\maketitle
\footnotetext [1]{{Dipartimento di Ingegneria Meccanica ed Aeronautica, Universit\`a di Roma La Sapienza, Via Eudossiana 18, 00184, Rome, Italy}}
\footnotetext[2]{{CNIS-Centro per le Nanotecnologie Applicate all'Ingegneria, Piazzale Aldo Moro, 5 00185 Roma, Italy}}
\footnotetext [3]{International Research Center M\&MOCS, Universit\`a dell'Aquila, Palazzo Caetani, Cisterna di Latina, (LT) 04012 Italy}
\footnotetext [4]{{Dipartimento di Ingegneria Meccanica e Aerospaziale, Universit\`a di Roma La Sapienza, Via Eudossiana,18, 00184, Roma, Italy} }
\footnotetext [5]{Diparimento di Matematica, Universit\`a di Roma La Sapienza, Piazzale Aldo Moro 5,  00185 Roma Italy}

\begin{abstract}
We consider the time evolution of a one dimensional $n$-gradient continuum. Our aim is to construct and analyze discrete approximations in terms of physically realizable mechanical systems, called microscopic because  they are  living on a smaller space scale.
We validate our construction by proving a convergence theorem of the microscopic system to the given continuum,  as the scale parameter goes to zero.
\end{abstract}

\section{Introduction}

Continua with exotic behaviors are acquiring an increasing attention for their interest in technological applications (see e.g. \cite{En,sep,Ma,Al,Pic,ZS}
and references therein).
In this paper we address what, in a sense, is an inverse problem: given a
continuum model we seek for those mechanical systems which, at a certain
length scale, behave as specified by the chosen continuum model. The aim is to understand the microscopic properties of such systems to obtain information on how to realize (synthesize) them, at least in principle.

To be more precise, we are interested in a metamaterial which, roughly speaking, is an array of elementary individuals, much smaller than the typical macroscopic size, arranged in periodic structures and exhibiting unusual macroscopic behavior.

In our mathematical analysis we want to consider such a continuous system as described by a partial differential equation generated by a Lagrangian which summarizes all the macroscopic properties we may desire.
Then we discretize this system  and manage to identify such a discretization as a real conservative mechanical model. In other words we start from a macroscopic behavior and describe one possible microscopic interaction which realizes it at a macroscopic level.  Finally we give a mathematical foundation to this procedure by proving a convergence result.

From a mathematical point of view, we underline once more that this is  an inverse problem, compared to the one (largely unsolved) formulated by D. Hilbert in his famous speach in 1900 at ICM in Paris (see \cite{H}) in which he encouraged to prove rigorously the transition from particle systems to fluid dynamics (Hilbert's 6-th problem). 
However it is worth to stress that we are working in the framework of continuum mechanics,  {but} our microscopic elements, even if small in macro unities, are large compared with molecular scales. 

We conclude this introduction by spending some more words on metamaterials, collocating them in the framework of generalized continua, with a particular emphasis to the pioneering work of G. Piola (see \cite{del14,P,Au}).

The rest of the paper is organized as follows. In Section 2 we introduce continuous and discrete Lagrangians and discuss the identification problem, namely we specify the mechanical systems outlined by the discretization procedure. In Section 3 we formulate and solve the associated convergence problem. 

We remark that our work concerns one-dimensional systems only. This is of course a severe limitation, but, on the other hand, it is a natural setting to start with. 

\subsection{Mechanical metamaterials}

By suitably rephrasing Engheta and Ziolkowski \cite{En} and Zouhdi et al. \cite{ZS}, metamaterials are materials which are first theoretically conceived and then
engineered to have properties very unlikely to be found in nature.

They are obtained by suitably assembling multiple individual elements
constructed with already available microscopic materials, but usually arranged
in (quasi-)periodic sub-structures. Indeed the properties of metamaterials do
not depend only on those of their component materials, {but also on 
 the topology of their connections and the nature of their mutual interaction forces.}
In literature it is currently specified a particular class of metamaterials,
so called mechanical metamaterials, those in which the particular properties
which are \textquotedblleft designed\textquotedblright\ for the newly
synthesized material are purely mechanical. The present paper deals exactly
with such a class.

We explicitly remark here that in the present paper we use the adjective
\textquotedblleft microscopic\textquotedblright\ or \textquotedblleft
micro-\textquotedblright\ meaning all those length scales which are (much) smaller
than the scale at which continuum mechanics is applicable. In particular we do
not attach any value in SI units to each considered length scale.

The particular shape, geometry, size, orientation and arrangement of
the elementary individual elements can affect, for instance, the propagation of waves
of light or sound in a not-already-observed manner. In this way one can
create material properties which cannot be found in conventional 
materials.

Particularly promising are those micro-structures which present high-contrast
in microscopic properties. These structures, once homogenized, have shown to
produce generalized continua (see e.g. \cite{Ch,Al,sep}). 
These micro-structures, although remaining quasi-periodical, are conceived so that some of the physical properties which
are characterizing their behavior are diverging when the size of the 
representative elementary volume tends to zero, while simultaneously some
others are vanishing in this limit.

To give a hint of the possible applications of newly designed metamaterials we
list here some among the papers which are more relevant to our results,
especially in the perspective of their extension to 2D and 3D systems.
In \cite{LPSW} it is shown how to synthesize a composite medium
exhibiting negative effective bulk modulus, negative effective mass density
(see also \cite{Ch}), or both properties. 
In \cite{K} materials with negative Poisson's ratio (auxetics) was
designed, and they were fabricated in 1999 (see Xu et al.\cite{XB} ). One of the most famous examples of such materials is the Goretex whose negative Poisson ratio opened unexpected possibility to e.g.
vascular surgery.

{The damping effects can be also suitably designed using special selection of the material microstructure as reported in \cite{a,b}, or the acoustic and optical effects as negative refraction, lensing and cloaking \cite{c,d}. }

All described materials can be modeled at a micro-level as finite dimensional
Lagrangian systems and their effective properties were all obtained via a kind
of homogenization procedure.

\bigskip

\subsection{Generalized Continua}

In the first half of XIX century the design of structures became an
intellectual activity based on the rigorous application of predictive
mathematical models. These models were formulated by means of a precise
postulation process and originated a series of problems or exercises directly
motivated by the engineering applications, which were solved by means of the
use of the then newly developed techniques of mathematical analysis.

The model describing the mechanical behaviour of materials introduced by
Cauchy - although very accurate for a large class of phenomena - cannot
be applied to all materials in every physical condition.

More general models were formulated by Gabrio Piola in the same years, but only recently they were considered in engineering for applications. 

In some formulations of continuum mechanics,  the possibility of the dependence of deformation energy on higher gradients of displacement, is rejected, due to an apparent (see \cite{dels95}) incompatibility with the second principle of thermodynamics (\cite{DS},\cite{Gurtin}). On the other hand, physicists, for instance  Landau \cite{L},  always considered this dependence as admissible, as they are
accustomed to base the postulation of physical theories on the principle of
least action or on the principle of virtual works, which is exactly the same starting point of G. Piola \cite{P}.

Actually, when introducing Piola continua, the true conceptual frame settled by
Cauchy, Navier and Poisson is to be drastically modified. The concept of
stress becomes secondary and the main role is played by deformation measures
together with action and dissipation functionals. The Euler-Lagrange equation
obtained in this more encompassing modeling process cannot be anymore
regarded to coincide with the balance of force unless one generalizes the
concept of force. This can be done by introducing generalized actions as the dual quantities in
the work of the gradients of displacements (see e.g. \cite{GR,MC,Mi,T1,Ge,del12,Ne}).

Actually the same concept of contact interaction has to be completely
modified, and the crucial point of determining the correct boundary conditions
which can be assigned in generalized continua theory has been addressed only
very recently (see e.g. \cite{del12}), following the original
ideas by Piola \cite{P}.

\section{Microscopic and Macroscopic descriptions 
}

In what follows we will consider  two length scales $l$ and $L$ with
$l\ll L.$ We will call microscopic or micro the description 
at the length scale $l$, while macroscopic or macro will be the attribute
relative to the description which is suitable at the lenght scale $L.$

We assume that the most suitable micro-description at micro-scale is ``discrete''  i.e. based on the model ``material particle'' (as done by Poisson,
Navier and -in some works- by Piola), while the description which has to be
used at the macro-level is that of a continuum, as introduced e.g. by
Lagrange, Cauchy or again Piola.

Remark however that we will not limit our attention to systems which verify
the assumptions put forward by Cauchy and Navier. We will
consider, actually, those continua which have been considered by Piola (and
then by many others, including Toupin, Green, Rivlin and Mindlin)  i.e. so
called higher gradient continua. 

To quantify the above considerations we will introduce, in the sequel, a small parameter $\e>0$ indicating the ratio between typical micro and macro scales, possibly to be sent to zero to outline a suitable asymptotic behavior.

\subsection{The basic macroscopic continuous model}

Let $I=(0,L)\subset \R$ be a finite interval assumed as reference configuration of the considered one-dimensional continuum. We label each element of the continuum with the coordinate $x\in I$ of  its placement in the reference configuration. The actual configuration of the continuum is described by the displacement field $u=u(x,t)$ which represents the horizontal displacement at time $t$ of the element $x$ from its position in the reference configuration.

Fixed an integer $n\ge 1$, for such a system we introduce  the Lagrangian 

\be\label{contlagr}
\mathcal{L}(u,\dot u)=\frac 1 2 \int_I |\dot u(x)|^2 -\int_I  \Phi(u(x), Du(x), \dots,D^nu(x)).
\ee
Here, $D^k u$
is the $k$-th $x$-derivative of $u$ and 	
\be \R^{n+1} \ni \underline{\x}=(\x_0, \dots,\x_n)\mapsto \Phi(\underline{\x})\in \R\ee
is a function whose properties will be specified later on.

Note that $\Phi(u,Du,\dots,D^nu)$ is the potential energy density corresponding to the dispacement $u$ and describes the constitutive properties of the medium under investigation.

The action on the time interval $(0,T)$  is consequently defined as
\be\label{contaction}\mathcal{A}=\int_0^T \mathcal{L}(u(\,\cdot\, ,t),{\dot u}(\,\cdot\, ,t)),\ee
where $\dot u(x,t)= \partial_tu(x,t)$ is the time derivative.

To deduce the Euler-Lagrange equations from the stationary action principle, we have first to specify the kinematic boundary condition for our problem. In the sequel we shall assume either
\begin{itemize}
\item periodic boundary conditions. Namely the reference configuration is $\mathcal{C}$, a circle of radius $\frac L{2\pi}$ (the points $0$ and $L$ are identified),

or

\item Dirichlet boundary conditions. Namely $u$ and its first $n-1$ derivatives vanish at $0$ and $L$.
\end{itemize}

With above boundary conditions no boundary terms appear when performing the integrations by parts needed to obtain the equation of the motion (\ref{condeq1}) below.

Note also that the maximal order of the spatial derivatives appearing in the equation of motion (\ref{condeq1})
is $2n$.

The equation of motion, as a consequence of the stationary action principle and the boundary conditions, is (with $D^0 u=u$)
\be\label{condeq1} \ddot u=-\sum_{\a=0}^n(-1)^\a D^\a\partial_{\x_\a}\Phi(u,Du,\dots,D^n u).\ee

We could also include, in the present context, a given external potential with a very minor effort. We avoid to do so for notational simplicity.

\bigskip

Now we specify $\Phi$ by assuming that 
\be\label{poten0}
\Phi(\underline{\x})=\frac 1 2 (\underline{\x},Q\underline{\x})+R(\underline{\x})
\ee
i.e. the quadratic part of $\Phi$ is a quadratic form in terms of the displacement and its derivatives, contained in the vector $\x$. 
$Q=\{Q_{\a,\b}\}_{\a,\b=0}^n$ is a symmetric (without loss of generality) constant matrix with $Q_{n,n}\neq 0$.  

On the non-linear part $R$ we shall do suitable assumptions later on. We start by requiring that
\be\label{assR0}R(\underline{0})=0, \quad R(\underline{\x})=O(|\x|^3),\ee
 namely the quadratic part of the interaction is fully expressed by the matrix $Q$. 
 
The fact that $\Phi$ is not depending explicitly on $x$ is consequence of the macroscopic homogeneity of the continuum (although it may be strongly inhomogeneous at microscopic scales). This implies that $Q$ is constant.

As a first step we show that, in contrast with the fairly generality of the model, the quadratic part can be considerably simplified. Indeed, symmetrizing, integrating by parts and using the periodic or Dirichlet boundary conditions, we get:
\begin{eqnarray}\label{ridQ}\mathcal{U}&\doteq&\frac 1 2 \sum_{\a,\b=0}^nQ_{\a,\b}\int _I D^\a u D^\b u\notag\\
&=&\frac 1 4 \sum_{\g=0}^{2n}\sum_{\substack  {\alpha, \beta\ge0 : \\ \alpha +\beta=\gamma} }
Q_{\a,\b} \int_I uD^\g u\big[(-1)^\a+(-1)^\b\big]\notag\\
&=&\frac 1 4 \sum_{\g=}^{n}\sum_{\substack  {\alpha, \beta\ge 0 : \\ \alpha +\beta=2\gamma} }
Q_{\a,\b} \big[(-1)^\a+(-1)^\b\big](-1)^\g\int_I |D^\g u|^2\\\notag
&=&\frac 1 2 \sum_{\g=0}^n A_\g \int_I|D^\g u|^2,\notag
\end{eqnarray}
where
\be\label{defA}A_\g=\frac 1 2 \sum_{\substack  {\alpha, \beta\ge 0 : \\ \alpha +\beta=2\gamma} }
Q_{\a,\b} \big[(-1)^\a+(-1)^\b\big](-1)^\g.\ee

{Note that in the first step in (\ref{ridQ}) we have used the symmetry of $Q_{\a,\b}$ and in the second step  we used that $(-1)^\a+(-1)^\b=0$ if $\a+\b$ is odd.  In the third step we have again integrated by parts.}

As a consequence of this analysis, without loss of generality, we can assume $\Phi$ of the form
\be\label{potenergyf}\Phi=\frac 1 2 \sum_{\a=0}^n A_\a |\x_\a|^2 +R(\underline{\x}),\ee
with $A_n\ne 0$ {and the equations of motion are
\be\label{condeq3} \ddot u=-\sum_{\a=0}^n(-1)^\a A_\a\Delta^\a u-\sum_{\a=0}^n(-1)^\a D^\a\partial_{\x_\a}\R(u,Du,\dots,D^n u),\ee
where $\Delta=D^2$ denotes the Laplacian. Note that in the linear part only even derivatives are allowed. }

\subsection{Formal discretization}
In view of the construction of the mechanical (microscopic) system with a finite number of degrees of freedom, we introduce a finite lattice of mesh $\e$ in $I$. The lattice points are $\{0,\e,2\e\dots,k\e,\dots N\e\}$ with the obvious condition $N\e=L$. When considering periodic boundary conditions we clearly identify $0$ with $\e N$. 

We associate to each lattice point a {\it microscopic} particle of unitary mass labelled by the index $i\in\{0,\dots,N\}$ and denote by $u_i$ the displacement of the particle $i$ from the reference position $i\e$. The array $u_\e=\{u_i\}_{i=0}^N$ is the discretized displacement field.

The discretized Lagrangian takes the form
\be\label{discrLagra0} \mathcal{L}_\e(u_\e,\dot u_\e)=\frac 1 2 \sum_{i=0}^N\e \dot u_i^2-U(u_\e),\ee
where
\be\label{discrpot}U(u_\e)=\sum_{i=0}^N\e \Big[\frac 1 2 \sum_{\a=0}^n A_\a|(D_\e^\a u_\e)_i|^2 +R((\underline{D_\e u_\e})_i)\Big],\ee
where $(\underline{D_\e u_\e})_i=\{(D^\a_\e u_\e)_i\}_{\a=0}^n$, 
\be\label{defDn}
D_\e^\a u= \begin{cases}
\Delta_\e^{\frac\a 2} u_\e,   &\a \text{ even},
\\\\
D_\e^+\Delta_\e^{\frac{\a-1}2} u_\e,  &\a\text{\ odd} . 
\end{cases}\ee
Here $D^+_\e$ and $D^-_\e$, defined as
\be\label{defDe}(D^+ u_\e)_i=\frac{u_{i+1}-u_i}{\e},\quad (D^- u_\e)_i=\frac{u_{i}-u_{i-1}}{\e},\ee
are the right and left discrete  derivatives respectively and $\Delta_\e$, defined by 
\be\label{discrlapla}
 (\Delta_\e u_\e)_i= (D^+_\e D^-_\e u_\e)_i=(D^-_\e D^+_\e u_\e)_i=\frac 1 {\e^2}(u_{i+1}+u_{i-1}-2u_i).\ee
is the discrete Laplacian.

To complete the above definitions we need to define the discrete derivatives at the boundary. For periodic boundary conditions it is enough to use the following convention: for any $k\in \Z$, 
\be\label{percond}u_{N+k}=u_k.\ee

For Dirichlet boundary condition, 
we have to think of the first and last $n$ particles frozen in their reference position.
Hence we assume the constraints
\begin{equation}\label{constr} u_i=0, \quad i\in\{0,\dots,n-1\}\cup\{N-n+1,\dots, N\}.\end{equation}
%
The equations of motion are
\be\label{newt} \ddot u_i= F_i, \quad F_i=-\frac{\partial U}{\partial u_i},\ee
with the index $i$ running from $1$ to $N$ in the periodic case and on the set of $i$'s for which $u_i$ is not constrained in the Dirichlet case.
 We notice that the choice of the right derivative (as well as any other possible discretization) is arbitrary. The only restriction that we have is the mechanical realizability (in principle) of this system. We are going to discuss this point in the next subsection.

We finally remark that $F_i$ depends on $u_j$, with $|i-j|\le n$. However this is an almost local contribution because $n$ is fixed and those $u_j$'s influencing $u_i$ are at macroscopic distance $O(\e)$.

\bigskip
\subsection{Realizable syntheses}

The aim of this subsection is to show that, at least in the simplest case of linear forces, the above introduced discrete system corresponds to a system of particles interacting via two-body forces of  range not larger than $n$. Therefore, it can be  realized by suitably assembling mechanical elements.

Let us consider the linear system  introduced in (\ref{condeq1}) with $R=0$ and its discrete counterpart (\ref{newt}).
It can be checked that
\be\label{F}F_i=- \sum_{k=0}^{n} (-1)^kA_{k}\D_\e^k u_\e(x_i).\ee
Therefore, the force acting on the particle $i$ is expressed as a linear combination of discrete derivatives up to the order $2n$. 

We want to show that $F_i$ can be interpreted as the result of the action of a system of linear pairwise forces with suitable range. More precisely, we want to find $\e$-dependent coefficients $k_{i,j}$ such that 
\be \label{Fe} F_i=\sum_j k_{i,j}(u_j-u_i)\ee
and hence
\be U(u_1,\dots ,u_N)= \frac 1 2\sum_{i,j=1}^N k_{i,j}(u_i-u_j)^2.\ee

We prove below that for  any $p$, 
\begin{equation}\label{recur} (\D^pu)_i=\sum_j K^p_{i,j}(u_j-u_i),\end{equation}
with $K^p_{i,j}$ other suitable constants.
Once (\ref{recur}) is proved, we can conclude that (\ref{Fe}) holds with
\be\label{kij} k_{i,j}=\sum_{p=0}^n (-1)^pA_pK^p_{i,j}.\ee

Note that the  constants $k_{i,j}$ are not necessarily all positive 
even if the $A_{\alpha}$ are all positive.

The constants $K^p_{i,j}$ are given by the recursive equation (\ref{Kij}) below. It implies that, for any $p$, $K^p_{i,j}$ vanishes for $|i-j|>p$, thus $k_{i,j}=0$ if $|i-j|>n$. Moreover, in the periodic case $K^p_{i,j}$ depends only on the difference $i-j$ and is symmetric in the exchange $i\leftrightarrow j$ and hence  the {\it action-reaction} principle is satisfied.

We  prove  (\ref{recur}) by recurrence. 

For $p=1$, we have
\be\label{D02}(\D_\e u)_i=\e^{-2}(u_{i+1}+u_{i-1}-2u_i)= \e^{-2} (u_{i+1}-u_i) +\e^{-2} (u_{i-1}-u_i).\ee
Thus (\ref{recur}) is verified with 
\be \label{D2}K^{1}_{i,i+1}=K^{1}_{i,i-1}=\e^{-2}\quad \text { and } K^1_{i,j}=0 \text{ otherwise}.\ee
Suppose now that (\ref{recur}) is true for $p= \ell-1$:
$$( \D_\e^{\ell-1}u)_i=\sum_j K^{\ell-1}_{i,j}[u_j-u_i].$$
Then, 
\begin{multline}(\D_\e^{\ell} u)_i
=(\D_\e^{\ell-1}\D_\e u)_i=\sum_j K^{\ell-1}_{i,j}[(\D_\e u)_j- (\D_\e u)_i]\\
=\sum_j K^{\ell-1}_{i,j}[\e^{-2} (u_{j+1}-u_j) +\e^{-2} (u_{j-1}-u_j)-\e^{-2} (u_{i+1}-u_i) -\e^{-2} (u_{i-1}-u_i)]\\=\sum_j K^{\ell-1}_{i,j}[\e^{-2} (u_{j+1}-u_i) - \e^{-2}(u_{j}-u_i)+\e^{-2} (u_{j-1}-u_i)-\e^{-2}(u_{j}-u_i)\\-\e^{-2} (u_{i+1}-u_i) -\e^{-2} (u_{i-1}-u_i)].
\end{multline}
Using the change of index $j+1\to j$ in the first term and $j-1\to j$ in the second, we have
\begin{multline}(\D_\e^{\ell}u)_i =\sum_j K^{\ell-1}_{i,j-1}\e^{-2} (u_{j}-u_i) - K^{\ell-1}_{i,j}\e^{-2}(u_{j}-u_i)+K^{\ell-1}_{i,j+1}\e^{-2} (u_{j}-u_i)\\-\e^{-2}K^{\ell-1}_{i,j}(u_{j}-u_i)-\e^{-2} K^{\ell-1}_{i,j}(u_{i+1}-u_i) -\e^{-2}K^{\ell-1}_{i,j} (u_{i-1}-u_i)].
\end{multline}
Thus, (\ref{recur}) is verified with the following recursive definition of $K^{\ell}_{i,j}$:
\be\label{Kij}K^{\ell}_{i,j}=\e^{-2}\Big[K^{\ell-1}_{i,j-1}+K^{\ell-1}_{i,j+1}-2K^{\ell-1}_{i,j}-(\delta_{i+1,j}+\delta_{i-1,j})\sum_{j'}K^{\ell-1}_{i,j'}\Big],\ee
for $\ell>1$ and $K^1_{i,j}$ given by (\ref{D2}).

{Equations (\ref{Kij}) and (\ref{kij}) solve definitely the posed problem of identifying the topology of the microstructure connections, since they provide the  coefficients $k_{i,j}$ only in terms of the coefficients $A_p$ that characterize the continuous formulation of the macroscopic description of the elastic problem.}

\section{A rigorous result of convergence\label{Mathematical result}}

{In this section we prove a convergence result of the discrete model introduced in the previous section to the prescribed continuous systems in the limit as the scale parameter  goes to $0$.  We show the convergence of the solution of the discrete system to the continuous one in the energy norm of the system.
To clarify the argument without the use of cumbersome notation, we present first a paradigmatic case  for which we discuss both periodic  and Dirichlet boundary conditions. 
The more general case is considered in Subsection 3.2 where we give the convergence proof only in the periodic case although the argument can be straightforwardly extended to the Dirichlet boundary conditions as well.}

For the reader convenience we rewrite the Lagrangian we are going to consider in this Section, namely

\be
\label{contlag0}
\mathcal{L}(u,\dot u)=\frac 12 
\int_I dx  | \dot u(x,t)| ^2 - \frac 1 2  \sum_{\alpha=1}^n \int_I dx  | D ^\alpha u |^2 (x,t))^2 -\int_I R(u, Du, D^2 u \dots ).
\ee
As we shall see later on, we will con sider only nonlinear terms $R$ depending on $u$ and the first derivative only.

\subsection{The $\Delta^2$ case - dynamic Euler-Bernoulli beam: {``Elastica''}}\label{elast}

\subsubsection{Periodic boundary conditions}\label{elastper}
We consider the Lagrangian (\ref{contlag0}) with $A_0=A_1=0$ and $A_2=1$. Moreover we focus on the linear case $R=0$.  Thus we have the following linear initial value problem in the circle, $\mathcal{C}$:
\begin{equation}
\label{em}
\ddot u=-\frac {\pa ^4 u}{\pa x^4} :=-\Delta^2 u,
\end{equation}
\be\label{incond}
u(x,0)=u_0(x),\quad \dot u(x,0)=v_0(x).
\ee
It is well known that there exists a unique classical solution  as the initial data are assumed sufficiently smooth.

More precisely we assume that
\be
\label{ass0}
u_0 \in H^s, \quad v_0 \in H^r \quad \text {with} \quad s\geq 6, r\geq 4,
\ee
where $H^s$ denotes the Sobolev space endowed with norm
$$
\| u\|_{H^s} = \sum_{\ell=0}^s \| D^\ell u \|_2^2,
$$
and $\|\,\cdot\,\|_p$ is the $L^p(\mathcal{C})$-norm.

In this way, by using the well known energy method, we can prove the propagation (in time) of the $H^s$ regularity for $u$  and $\dot u$, yielding, in particular,   $u \in C^5(\mathcal{C})$ (as consequence of the obvious inequality $ \| u\|_{\infty} \leq C \| u \|_{H^1}$). 

Next we consider the mechanical system of $N$ particles, with coordinates $u_i$,  $i=1,\dots N$, whose Lagrangian is given by (\ref{discrLagra0}) again with  $A_0=A_1=0$, $A_2=1$ and $R=0$.
The equation of motion are explicitly 
\begin{equation}
\label{em1}
\ddot u_i=\frac 1 {\e^4} (-u_{i+2}+4 u_{i+1} -6u_i -u_{i-2}+4 u_{i-1})
\qquad i=1 \dots N,
\end{equation}
with the convention $u_{N+k}=u_k$ for any $k\in \Z$.

We want to compare the solutions of (\ref{em}) with the corresponding ones of  (\ref{em1}). To do this we first set
\begin{equation}\label{step}
u_\e (x,t)=u_i(t) \quad \text {if} \quad x\in [i\e, (i+1) \e ),\quad i\in\{1,\dots, N\}.
\end{equation}
 In other words we introduce a function $u_\e$ which is the step, left continuous, function (constant in the lattice interval) taking the value of the nearest left point of the lattice. Problem (\ref{em1}) is rephrased accordingly:
\begin{equation}
\label{em2}   
\ddot u_\e (x,t) =-\Delta_\e ^2u_\e(x,t)
\qquad x\in \mathcal{C},
\end{equation}
where
\begin{equation}\label{Delta}
\Delta_\e u(x)=D_\e^+D_\e^- u(x)
\end{equation}
\begin{equation}\label{D}
D^{\pm}u(x)=\pm \frac 1{\e} (u(x\pm\e)-u(x)).
\end{equation}
Notice that the Lagrangian (\ref{discrLagra0}), with $A_0=A_1=0$, $A_2=1$ and $R=0$, has the following continuous representation:
\begin{equation}
\label{conlagep}
\mathcal{L}(u_\e,\dot u_\e)=
\int_{\mathcal{C}} dx \Big[\frac 1 2 \dot u_\e(x,t)^2 - \frac 1 2 (\Delta_\e u_\e(x,t))^2\Big].
\end{equation}

We suppose that, at the initial time,  $u_\e, \dot u_\e$ are approximating  $u, \dot u$ in the sense that 
\begin{equation}
u_\e(x,0)=u_0(i\e) ,\quad \dot u_\e (x,0)=v_0 (i\e) \quad \text {if} \quad x\in [i\e, (i+1)\e ).
\end{equation}
Note that, by the conservation of the energy, we have
\be\label{EC}\mathcal{E}[u(t)]:=\frac 12\int_{\mathcal{C}}dx \Big[|\dot u(t)|^2 + |\Delta u(t)|^2\Big]=\mathcal{E}[u(0)], \ee
as well as
\be\label{ECD}\mathcal{E}_\e[u_\e(t)]:=\frac 12\int_{\mathcal{C}}dx \Big[|\dot u_\e(t)|^2 + |\Delta_\e u_\e(t)|^2\Big]= \mathcal{E}_\e[u_\e(0)]. \ee

Next we introduce the following function which controls the deviation of $u_\e$ from $u$:
\begin{equation}
\label{W}
W_\e (t)= \frac 12 \int_ \mathcal{C}dx\Big[(u_\e (x,t)- u (x,t))^2+ (\dot u_\e (x,t)-\dot u (x,t))^2+[\Delta_\e ( u_\e (x,t)- u (x,t))]^2\Big].
\end{equation}
Computing the time derivative and using the equation of motion we get
\begin{eqnarray}\label{wdot}
\dot W_\e(t)=&& \int_\mathcal{C} dx (\dot u_\e (x,t)-\dot u (x,t)) (u_\e (x,t)- u (x,t)+\ddot u_\e (x,t)-\ddot u (x,t))+\nn\\ \nn
&&
\int_\mathcal{C} dx  \Delta_\e ( u_\e (x,t)- u (x,t)) \Delta_\e ( \dot u_\e (x,t)-\dot u (x,t))= \\ 
&&
+\int_\mathcal{C} dx (\dot u_\e (x,t)-\dot u (x,t))(u_\e (x,t)- u (x,t))\nn\\&&
-\int_\mathcal{C} dx (\dot u_\e (x,t)-\dot u (x,t)) \Delta_\e^2( u_\e (x,t)- u (x,t))\big]\\ \nn
&&
+\int_\mathcal{C} dx (\dot u_\e (x,t)-\dot u (x,t)) (\Delta^2 u (x,t)- \Delta_\e^2 u (x,t))\\ \nn
&&
+\int_\mathcal{C} dx  \Delta_\e ( u_\e (x,t)- u (x,t)) \Delta_\e ( \dot u_\e (x,t)-\dot u (x,t)) \nn.
\end{eqnarray}
Now consider the following discrete integration by parts formula, namely
\be\label{intbypartsper}
\int_\mathcal{C}  f(x) D^\pm_\e g(x) = -\int_\mathcal{C} D^\mp_\e f(x)  g(x)
\ee
valid for any couple of bounded functions $f$ and $g$. 

 If we apply the above formula twice
we conclude that the second and fourth terms in (\ref{wdot}) cancel each other. On the other hand the first term is bounded by 
\[\frac 1 2 \int_{\mathcal{C}} dx |\dot u-\dot u_\e|^2 +|u-u_\e|^2\le W.\]
The third term is bounded by
\[
\frac 1 2\int_\mathcal{C} dx (\dot u_\e (x,t)-\dot u (x,t))^2+
\frac 1 2\int_\mathcal{C}  |(\Delta^2 - \Delta_\e^2) u (x,t)|^2 .
\]
Now the first term of the expression above is bounded by  $W$.  The second one, by the regularity of $u$ and its derivatives up to the fifth order,  is bounded, uniformly in $x\in \mathcal{C}$ and in $t$ in any bounded interval, by a constant $\om_\e$ vanishing as $\e \to 0$.  Here and in the rest of the paper $\omega_\e\in \R$ denotes such a generic infinitesimal constant.

\noindent In conclusion, by the Gronwall lemma,
\be
W_\e (t) \leq W_\e (0)\text{\rm e}^{2t} +  \om_\e  t \text{\rm e}^{2t}
\ee
so that  $W_\e (t)$ is vanishing, because $W_\e(0)\to 0$ by the regularity of $u$ and the assumptions on initial data.

We summarize above discussion in the following 
\begin{theorem}\label{elasticaper}
Suppose that $u_0$ and $v_0$  satisfy \eqref{ass0}.  Let $u(t)$   be the solution to (\ref{em}) and $u_\e(t)$ be the step function defined by (\ref{step}) with $u_i(t)$, $i=1,\dots,N$, solutions to (\ref{em1}) with initial data $u_i(0)=u_0(i\e)$ and $\dot u_i(0)=v_0(i\e)$. Then, for any $t\in \R$, 
\[\lim_{\e\to 0} W_\e(t)=0.\]
\end{theorem}
\subsubsection{Dirichlet boundary conditions}\label{elastdir}
For the Dirichlet boundary conditions we replace the circle $\mathcal{C}$ with the interval ${I}=[0,L]$. The equation (\ref{em}) is well posed with the boundary conditions 
\begin{equation}
\label{bc}
u(0,t)=u'(0,t)=u(L, t)=u'(L,t)=0.
\end{equation}
{\bf Remark}: Several other boundary conditions may have an interest in engineering application and a physical meaning. For instance the conditions $u(0)=u''(0)=0$, $u(L)=u''(L)=0$, characterizes a beam with pivots applied at its endpoints, while the conditions which we considered here are relative to clamped-clamped beams. We do not consider in this paper the other possible boundary conditions, as the focus of this paper is different.

\bigskip
Again, by using the energy method, we can construct solution with $H^s$ regularity, by assuming
\be
\label{ass01}
u_0 \in H_0^2 \cap H^s, \quad v_0 \in H_0^2 \cap H^r \quad \text {with} \quad s\geq 6, r\geq 4.
\ee
Here $H^2_0$   (introduced to take into account the boundary conditions) is defined as the space of the $H^2$ functions vanishing in $0$ and $L$, together with their first derivative.

\bigskip

The corresponding discrete system is constituted by $N-3$ particles 
with coordinates $u_i$, $i=2, \dots,N-2$ and
\be \label{discbc}u_0=u_1=u_{N-1}=u_N=0\ee
are the constraints corresponding to the Dirichlet boundary conditions.  

With this position, the explicit equations of motion are 
\begin{equation}
\label{em1dir}
\ddot u_i=\frac 1 {\e^4} (-u_{i+2}+4 u_{i+1} -6u_i -u_{i-2}+4 u_{i-1})
\qquad i=2 \dots N-2,
\end{equation}
As before we introduce the left continuous step function
\begin{equation}\label{stepdir}
u_\e (x,t)=u_i(t) \quad \text {if} \quad x\in [i\e, (i+1) \e ),\quad i\in\{0,\dots, N-1\},\end{equation}
but we find convenient to think of it as a function on $\R$ extended with  value $0$ outside ${I}$.
Then (\ref{em1dir}) can be rewritten similarly to (\ref{em2}) as
\begin{equation}
\label{em2dir}
\ddot u_\e (x,t) =-\Delta_\e ^2u_\e(x,t)
\qquad x\in {I}_\e=(2\e, L-\e).\end{equation}
Note that the values of $u_i$ are frozen for $i=0,1,N-1,N$, so that $u_\e=0$ in ${I}^c_\e={I}-{I}_\e$.
We also think of the solution $u$ of the continuous equation as extended with value $0$ outside of ${I}$ 

Next we introduce the function $W_\e(t)$ as
\begin{equation}
\label{W1}
W_\e (t)= \frac 12 \int_ \R dx (\dot u_\e (x,t)-\dot u (x,t))^2+
\frac 12 \int_\R dx  [\Delta_\e ( u_\e (x,t)- u (x,t))]^2.
\end{equation}
Note that this function differs from the one defined by integrating on ${I}$ instead of $\R$  because $\Delta_\e$ is non-local. It is actually larger and hence provides a stronger control of the convergence.
Now we compute again the time derivative of $W$, as before and we get
\be\label{wdotdir}
\dot W_\e(t)= \int_\R dx (\dot u_\e (x,t)-\dot u (x,t)) (\ddot u_\e (x,t)-\ddot u (x,t))+\int_\R dx  \Delta_\e ( u_\e (x,t)- u (x,t)) \Delta_\e ( \dot u_\e (x,t)-\dot u (x,t)). 
\ee
By using twice the discrete integration by parts formula
\[\int_\R dxfD^\pm_\e g=-\int_\R gD^\mp_\e f ,\]
valid of any couple of bounded compactly supported functions $f$ and $g$,  the second term becomes, as before
\[\int_\R dx  ( \dot u_\e (x,t)-\dot u (x,t))\Delta_\e^2 ( u_\e (x,t)- u (x,t)) .\]
As for the the first term, we need to use the equations of motion (\ref{em}) for $u$ and (\ref{em2dir}) for $u_\e$. Note that the last ones hold only in ${I}_\e$. Thus,
using that $\ddot u_\e=0$ in $\R-{I}_\e$, 
the first term becomes 
\begin{multline}-\int_{{I}_\e} dx (\dot u_\e (x,t)-\dot u (x,t)) (\Delta^2_\e u_\e (x,t)-\Delta^2 u (x,t))-\int_{\R-{I}_\e}(\dot u_\e(x,t)-\dot u(x,t))(-\Delta^2u(x,t))=\\
-\int_\R dx (\dot u_\e (x,t)-\dot u (x,t)) (\Delta^2_\e u_\e (x,t)-\Delta^2 u (x,t))+\int_{\R-{I}_\e}(\dot u_\e(x,t)-\dot u(x,t))\Delta_\e^2u_\e(x,t)).
\end{multline}
By adding and subtracting the term $\int_{\R} dx (\dot u_\e (x,t)-\dot u (x,t)) (\Delta^2_\e u_\e (x,t)-\Delta_\e^2 u (x,t))$ the above term becomes
\begin{multline}-\int_{\R} dx (\dot u_\e (x,t)-\dot u (x,t)) (\Delta^2_\e u_\e (x,t)-\Delta_\e^2 u (x,t))-\int_{\R} dx (\dot u_\e (x,t)-\dot u (x,t)) (\Delta_\e^2 u(x,t)-\Delta^2 u (x,t))\\+\int_{\R-{I}_\e}(\dot u_\e(x,t)-\dot u(x,t))\Delta_\e^2u_\e(x,t)).\end{multline}
Putting together all these terms we conclude that
\be\label{Wdotf}\dot W=-\int_{\R} dx (\dot u_\e (x,t)-\dot u (x,t)) (\Delta_\e^2 u(x,t)-\Delta^2 u (x,t))+\int_{\R-{I}_\e}(\dot u_\e(x,t)-\dot u(x,t))\Delta_\e^2u_\e(x,t))\ee
The first term in the right hand side of (\ref{Wdotf}) goes to $0$ as in the periodic case, by the regularity $u$. 
The second term is the novelty of the Dirichlet case.
In order to estimate it, note that $\dot u_\e=0$ outside of ${I}_\e$, hence we need to estimate $\int_{\R-{I}_\e}\dot u(x,t)\Delta_\e^2u_\e(x,t))$. 

By the boundary conditions on $u$, ($u=0$ and $u'=0$ in $0$ and $L$ for any $t$),  it results 
(by our assumptions $|\Delta\dot u(x,t) |$ is bounded) 
\be\label{u''}|\dot u(x,t)|\le \frac 1 2 \sup_{x\in \mathcal{I}} |\Delta\dot u(x,t) |\e^2\quad x\in {I}_\e^c.\ee
Furthermore
\[\Delta^2_\e u_\e(x)=\e^{-2}( \Delta_\e u_\e(x+\e)+\Delta_\e u_\e(x-\e)-2\Delta_\e u_\e(x)).\]
Rewriting the total energy (\ref{ECD}) in a more explicit form,
\be \mathcal{E}[u_\e]=\frac 12 \sum_{i=2}^{N-2} \e |\dot u_\e (\e i)|^2 + \frac 12 \sum _{i=1}^{N}\e 
|\Delta_\e u_\e(i\e)|^2,
\ee
we obtain, at any time and for any $x$ in $\mathcal{I}_\e$,
\be\label{boundLapl}|\Delta_\e u_\e(x)|\le \frac{\sqrt{2 E_0}}{\sqrt\e},\ee
where $E_0=\mathcal{E}(u(0))$ is the energy of the initial data.
Hence 
\be\label{delta2}\sup_{x\in {I}} |\Delta^2_\e u_\e(x)|\le {4 \sqrt{2E_0}}{\e^{-\frac 5 2}}.\ee
Combining (\ref{u''}) and (\ref{delta2})
and using  the fact that the integration is restricted to the set  
${I}-{I}_\e$, whose measure is $4\e$ (remind that  $\dot u=0$ outside $\mathcal{I}$), we conclude that 
\[\Big |\int_{\R-{I}_\e}\dot u(x,t)\Delta_\e^2u_\e(x,t))\Big|\le C\sqrt{\e}.\]
 
The rest of the argument proceeds as before and we conclude that $W_\e(t)\to 0$.

We summarize above discussion in the following
\begin{theorem}\label{elasticadir}
Suppose that $u_0$ and $v_0$  satisfy \eqref{ass01}.  Let $u(t)$   be the classical solution to (\ref{em}) with boundary conditions (\ref{bc}) and initial values (\ref{incond}) and $u_\e(t)$ be the step function defined by (\ref{stepdir}) with $u_i(t)$, $i=2,\dots,N-2$, solutions to (\ref{em1dir}) with initial data $u_i(0)=u_0(i\e)$ and $\dot u_i(0)=v_0(i\e)$. Then, for any $t\in \R$, 
\[\lim_{\e\to 0}W_\e(t)= 0.\]
\end{theorem}

\subsection{A  $n$-th gradient case}
Now we extend the previous argument to the more general setup corresponding to the Lagrangian (\ref{contlag0}),  restricting  the discussion  to the simpler case of periodic boundary conditions. The Dirichlet boundary conditions can be handled as in the previous subsection but we  avoid here unnecessary  complications.

We assume  the  following conditions:
\begin{enumerate}
\item
\be \label{ann}A_{0}>0,\quad A_n>0,\quad A_\a\ge 0, \quad \a=1,\dots,n-1\ee
\item  We have already supposed that $R(\underline{0})=0$ and $R(\underline{\x})=O(|\underline{\x}|^3)$. In addition we assume that, for $n=1$, $R$ depends only on $u$ and, for $n\ge 2$, $R$ depends only on $u$ and $Du$.
Moreover  we assume  $R \in C^{2n+2}(\R^2)$. 
\end{enumerate}

\bigskip
\noindent{\bf Remark 3.1:}  The positivity assumptions on the $A_\a$'s with $\a=1,\dots,n-1$, can be relaxed. 
In facts, 
let us define, for some $\e_0 >0$,
\be 
\kappa=\sup_{\e \in (0,\e_0) }\sup_{u: \|D_\e u\|_2\le 1}\frac {\|u\|^2_2}{\|D_\e u\|^2_2},
\ee
with the supremum on $u$ taken on all $u$ with $0$ average. Then it is enough to assume
\be\label{assa1}
\sum_{\substack  {\alpha=1 : A_\a<0}}^{n-1} |A_{\alpha}|\kappa^{n-\a}\le \frac 1 2 A_{n},
\ee
to make the argument of the proof   still working.
This remark allows us to consider, for instance, the case $\ddot u= (-\Delta ^2-\gamma \Delta)u $, with $\gamma$  sufficiently small, excluded by (\ref{ann}).
\medskip

\noindent{\bf Remark 3.2:} The assumption on $R$ concerning its   dependence  on $u$ and $Du$ only, is restrictive. We do not expect any surprise in assuming an explicit dependence on some higher derivatives.  However, as we shall see in the course of the proof,  more general assumptions would complicate the algebraic manipulations in dealing with the discrete derivatives  in a consistent way.

As regards the initial data we assume
\be
\label{ass02}
u_0 \in H^{2n+2} , \quad v_0 \in H^{n+2},
\ee
and, as before, the $H^s$ regularity is propagated. Clearly  $u  \in C^{2n+1}(\mathcal{C})$.

\bigskip
The explicit equation is 
\begin{equation}\ddot u+\sum_{\alpha=0}^n  (-1)^\a A_{\alpha}\D^\a u+\partial_{\x_0}R(u,Du)-D[\partial_{\x_1}R(u,Du)]=0.\label{conteq}\end{equation}

Note that thanks to  the energy conservation,
\be\label{conten}\mathcal{E}[u]=\int_{\mathcal{C}} dx 
\Big[\frac 1 2 \Big\{ \dot u^2+\sum_{\alpha=0}^n A_{\alpha}| D^\alpha u|^2  \Big\} + R(u,Du),\Big]
\ee
we get immediately an a priori bound on the $L^2$ norm of $u$, $\dot u$ and $D^n u$:
\be      
\label{contenerbound}
\frac 1 2 \int_{\mathcal{C}}dx \Big[ |\dot u|^2+A_{0}|u|^2+ A_{n}|D^n u|^2\Big]\le \mathcal{E}[u(0)].
\ee

\bigskip
Now we remind the discrete counterpart of the above setup, which corresponds to  the discrete Lagrangian (\ref{discrLagra0}). Using the discontinuous function  
\begin{equation}
u_\e (x,t)=u_i(t) \quad \text {if} \quad x\in [i\e, (i+1) \e ),
\end{equation}
as in the previous section, the discrete Lagrangian can be written as
\be\label{genlagr}
\mathcal{L}_\e=\int_{\mathcal{C}} dx\Big[\frac 1 2 |\dot u_\e(x,t)|^2-\frac 1 2 \sum_{\a=0}^n A_\a |D_\e^\a u_\e(x,t)|^2 - R(u_\e(x,t), D^+_\e u_\e(x,t))\Big].
\ee

We can write the associated equations of motion in terms of $u_\e$ as
\be\label{discreq}
\ddot u_\e+\sum_{\alpha=0}^n  (-1)^\alpha A_{\alpha}\D_\e^{\alpha} u_\e+ \partial_{\x_0} R(u_\e, D^+_\e u)-D^-_\e[\partial_{\x_1} R(u_\e,D^+_\e u_\e)]=0.
\ee

Also for the discrete system the energy conservation holds. Thus  we have
that 
\be\label{enerdiscr}\mathcal{E}_\e[u_\e]=\int_{\mathcal{C}} dx 
\Big[\frac 1 2 \Big\{ \dot u_\e^2+\sum_{\alpha=0}^n A_{\alpha}|D_\e^\alpha u|^2\Big\} + R(u_\e,D^+_\e u_\e),\Big]
\ee
is conserved and hence,  using  that  $R\ge 0$,  we have the inequality
\be\label{discrenerbound}\frac 1 2 \int_{\mathcal{C}}dx \Big[ |\dot u_\e|^2+A_{0}|u_\e|^2+ A_{n}|D_\e^n u_\e|^2\Big]\le \mathcal{E}_\e[u_\e(0)].\ee
Since $A_{0}>0$ and $A_{n}>0$, the existence, globally in time, for the solution to the discrete system follow from this bound.

\bigskip

We start by proving the convergence of the discrere system to the continuous one in the linear case, namely when $R=0$,  
\be\label{conteq1}\ddot u+\sum_{\alpha=0}^n  (-1)^\a A_{\alpha}\D^{\alpha} u=0.\ee

Similarly, the discrete system becomes
\be\label{discreq1}\ddot u_\e+\sum_{\alpha=0}^n  (-1)^\a A_{\alpha}\D_\e^{\alpha} u_\e=0.
\ee

We introduce
\begin{multline}\label{W2*}
W_\e(t)=\frac 12 \int_{\mathcal{C}} dx \Big\{ |u(x,t)-u_\e(x,t)|^2+ |\dot u(x,t)-\dot u_\e(x,t)|^2 \Big\}+  \int_{\mathcal{C}} dx\sum_{\a=0}^nA_{\a} |D^\a_\e [u(x,t)-u_\e(x,t)] |^2 .
 \end{multline}

The time derivative of $W$ is:
\begin{eqnarray}
\frac{d}{dt}W_\e&&=\int_{\mathcal{C}} dx
\Big\{
(\dot u-\dot u_\e) ( u-u_\e+\ddot u-\ddot u_\e) +\sum_{\a=0}^n A_\a D_\e^\a(\dot u- \dot u_\e) D^\a_\e (u-u_\e)
\Big\}\nn
\\&&=
\int_{\mathcal{C}} dx   (\dot u-\dot u_\e) \Big\{ ( u-u_\e+\ddot u-\ddot u_\e) +\sum_{\a=0}^n (-1)^\a A_\a \D^\a_\e (u-u_\e)\Big\}\nn
\\&&=
\int_{\mathcal{C}} dx(\dot u-\dot u_\e)\Big\{( u-u_\e) +\sum_{\a=0}^n (-1)^\a A_\a \Big[
\D^\a_\e u_\e -\D^\a u -\D^\a_\e u_\e+\D^\a_\e u
\Big]\Big\}
\\&&=
\int_{\mathcal{C}} dx(\dot u-\dot u_\e)\Big\{ u-u_\e+ \sum_{\a=0}^n(-1)^\a A_\a (\D_\e^\a u-\D^\a u) \Big\}.\nn
\end{eqnarray}
In the second step we have integrated by parts $\alpha$ times, in the third we have used the equations of motion. The last step follows by canceling two equal terms with opposite sign.

Hence the linear case goes exactly as in previous subsection, because $|\D^\a_\e  u-\D^\a u|\le \omega_\e$ for $\a\le n$ and $u\in C^{2n+1}(\mathcal{C})$.

We summarize the results  for the linear case in the following 
\begin{theorem}
Assume  $R=0$ and 
suppose that $u_0$ and $v_0$ satisfy \eqref{ass02}.
Let $u(t)$   be the classical solution to (\ref {conteq}) and $u_\e(t)$ be the step function defined by (\ref{step}) with $u_i(t)$, $i=1,\dots,N$, solutions to (\ref{conteq1}) with initial data $u_i(0)=u_0(i\e)$ and $\dot u_i(0)=v_0(i\e)$. 

Then,
for any $t\in [0,T]$, 
\[\|u(t)-u_\e(t)\|_\e\to 0, \quad \text{ as }\e \to 0,\]
where $\| \,\cdot\,\|_\e$ is the $\e$-dependent  norm defined by 
\be\label{norma} \|u\|_\e^2=\int_{\mathcal{C}} dx \Big\{|\dot u|^2+|u|^2+ \sum_{k=1}^n |D_\e^k u|^2\Big\}.\ee\end{theorem}

\bigskip

Next we consider the nonlinear case. 
Now the equations of motion are (\ref{conteq}) and (\ref{discreq}) for the continuous and discrete system respectively.  Defining $W_\e$ by \eqref{W2*},
by the same computation, we have, again using the summation by parts formula, 
\begin{eqnarray}\label{Wdotnon}
\frac{d}{dt}W_\e&&=\int_{\mathcal{C}} dx
(\dot u-\dot u_\e)\Big\{( u-u_\e+\ddot u-\ddot u_\e\Big\} +\sum_{\a=0}^n A_\a D_\e^\a(\dot u- \dot u_\e)(D^\a_\e (u-u_\e)
\nn
\\&&=
\int_{\mathcal{C}} dx(\dot u-\dot u_\e)\Big\{( u-u_\e+\ddot u-\ddot u_\e) +\sum_{\a=0}^n (-1)^\a A_\a \D^\a_\e (u-u_\e)\Big\}\nn
\\&&=
\int_{\mathcal{C}} dx(\dot u-\dot u_\e)\Big\{( u-u_\e) +\sum_{\a=0}^n (-1)^\a A_\a \Big[
\D^\a_\e u_\e -\D^\a u -\D^\a_\e u_\e+\D^\a_\e u
\Big]\\&&+\Big[-\partial_{\x_0}R(u,Du)+D\partial_{\x_1}R(u,Du)+\partial_{\x_0}R(u_\e,D^+_\e u_\e)-D^-_\e\partial_{\x_1}R(u_\e,D^+_\e u_\e)
\Big]\Big\}\nn
\\&&=
\int_{\mathcal{C}} dx(\dot u-\dot u_\e)\Big\{ u-u_\e+ \sum_{\a=0}^n(-1)^\a A_\a (\D_\e^\a u-\D^\a u) \nn
\\&&+\Big[-\partial_{\x_0}R(u,Du)+D\partial_{\x_1}R(u,Du)+\partial_{\x_0}R(u_\e,D^+_\e u_\e)-D^-_\e\partial_{\x_1}R(u_\e,D^+_\e u_\e)
\Big]\Big\}\nn
\end{eqnarray}

To control the non-linear terms we proceed by estimating:
\be
\label{T1}
T_1= \partial_{\x_0} R(u,Du)-\partial_{\x_0}R(u_\e, D^+_\e u_\e)=T_1^1+T_1^2
\ee
and
\be
\label{T2}
T_2=D^-_\e [{\partial_{\x_1} R}(u_\e,{D^+_\e }u_\e)]-{D}[{\partial_{\x_1} R}(u,{D}u )] =T_2^1+T_2^2
\ee
where 
\be
\label{T11}T_1^1
=
\partial_{\x_0} R(u,Du)-\partial_{\x_0}R(u, D^+_\e u),
\ee

\be
\label{T12}T_1^2=
\partial_{\x_0} R(u,D^+_\e u)-\partial_{\x_0}R(u_\e, D^+_\e u_\e),
\ee

\be
\label{T21}
T_2^1={D^-_\e}[{\partial _{\x_1}R}(u,{D^+_\e }u)]-{D}[{\partial_{\x_1} R}(u,{D }u)], 
\ee
and
\be
\label{T22}
T_2^2={D^-_\e}[{\partial_{\x_1} R}(u_\e,{D^+_\e }u_\e)]-{D^-_\e}[{\partial_{\x_1} R}(u,{D^+_\e }u)].
\ee

The bound (\ref{contenerbound}) and Poincar\'e inequality imply that the $L^\infty$ norms of $u$, $Du$ and $D^\pm_\e u$ are bounded uniformly in $\e$. Thus, by the local Lipschitz continuity of $\partial_{\x_0} R$, we have
\[| T_1^1| \le C|Du-D^+_\e u|\le \omega_\e,\]
by the regularity of $u$.
Thus, by the energy bounds (\ref{contenerbound}) and (\ref{discrenerbound}) we get
\[\Big|\int_{\mathcal{C}} dx(\dot u-\dot u_\e)T_1^1 \Big|\le C[\mathcal{E}(u(0))+\mathcal{E}_\e(u_\e(0))]^{\frac 1 2}\omega_\e,\]

To control $T_1^2$  we need $L^\infty$ bounds for $u_\e$ and $D_\e u_\e$. They follow from the conservation of the energy for the discrete system by means of the following 
\begin{lemma}\label{sobo}
Let $f$ be a step function on $\mathcal{C}$  left continuous in the points $i\e$. Suppose that 
\[\|f\|_{H_\e^1}^2=\int_{\mathcal{C}} dx( |f|^2+|D^+_\e f|^2)\]
is bounded. Then 
\[\|f\|_\infty\le  C\|f\|_{H_\e^1}.\]
\end{lemma}
\begin{proof}
Let $x_0=i_0\e $ be any point such that $|f(x_0)|^2\le \frac 1{|\mathcal{C}|}\int_{\mathcal{C}} dx|f|^2$. Note that such a point does exist otherwise we would obtain a contradiction ($\int_{\mathcal{C}}dx|f|^2>\int_{\mathcal{C}}dx|f|^2$). For $x=(i_0+k)\e$ we have
\[f^2(x)=f^2(x_0)+\sum_{h=0}^{k-1} [f^2(x_0+(h+1) \e)-f^2(x_0+h\e)].\]
Since 
\begin{multline*}|f^2(x+ \e)-f^2(x)|=|[f(x+ \e)+f(x)][f(x+\e)-f(x)]=\e[f(x+ \e)+f(x)]D^+_\e f|\\\le \frac 1 2 \e[f(x+ \e)+f(x)]^2+ \frac 1 2 \e|D^+_\e f|^2,\end{multline*}
we conclude that
\[|f^2(x)|\le \left(\frac 1{|\mathcal{C}|}+1\right)\int_{\mathcal{C}} dx|f|^2+\frac 1 2\int_{\mathcal{C}} dx|D^+_\e f|^2\le\left(\frac 1{|\mathcal{C}|}+1\right)\|f\|_{H_\e^1}^2.
 \]
\end{proof}

Lemma \ref{sobo} and the energy bound (\ref{discrenerbound}) imply that the $L^\infty$ norms of $u_\e$ and $D^+_\e u_\e$ are bounded uniformly in $\e$. Thus we can use the Lipschitz continuity of $\partial_{\x_0} R$ to get:
\[\Big|\int_{\mathcal{C}} dx(\dot u-\dot u_\e) T_1^2 \Big|\le K W_\e,\]
with $K$ the Lipschitz constant of $\partial_{\x_0}R$ in the ball of radius $\max\{\|u\|_\infty, \|D u\|_\infty,\|u_\e\|_\infty, \|D_\e u_\e\|_\infty\}$.

The bound of $T_2$, involving discrete derivatives, requires the following chain rule formula for the discrete derivative of a composite function:
\begin{lemma}\label{chain}
If $f$ has continuous first derivative $f'$, then for any function $g$ and for any $x$ there exist $\lambda_{\e,x}\in (0,1)$ such that
\[D^\pm_\e f(g(x))=f'(\zeta_\e(x))D^\pm_\e g(x), \quad \text{with }\zeta_\e(x)=g(x)+\e \lambda_{\e,x} D^\pm_\e g(x)\]
\end{lemma}
\begin{proof}By the mean value theorem, for $D^+_\e$ we have
\[D^+_\e f(g(x))=\e^{-1}[f(g(x+\e)-f(g(x)]=\e^{-1}\int_{g(x)}^{g(x+\e)}dz f'(z)=\e^{-1}[g(x+\e)-g(x)]f'(\zeta)\]
for a suitable $\zeta$ in the interval with extremes $g(x)$ and $g(x+\e)$: $\zeta=g(x)+\lambda_{\e,x}[g(x+\e)-g(x)]=g(x)+\e\lambda_{\e,x}D^+_\e g(x)$ for some $\lambda_{\e,x}\in (0,1)$. In the same way the statement for $D^-_\e$ follows.
\end{proof}

\bigskip
By the chain rule,
$$
T_2^1=
{\partial^2 _{\x_0,\x_1}R}(\zeta_\e(x),D^+_\e u)D^-_\e u-{\partial^2 _{\x_0,\x_1}R}(u,Du)Du+
{\partial^2 _{\x_1^2}R}(u,\eta_\e(x))\D_\e u-{\partial^2 _{\x_1}R}(u,Du)\D u,
$$
where 
\[\zeta_\e(x)=u(x)+\e\lambda_{\e,x} D^-_\e u(x)\] and \[\eta_\e(x)= D_\e u(x)+\e\mu_{\e,x}\D_\e u(x),\] 
with $\lambda_{\e,x}\in(0,1)$, $\mu_{\e,x}\in (0,1)$. But
\begin{multline*}
{\partial^2 _{\x_0,\x_1}R}(\zeta_\e(x),D^+_\e u)D^-_\e u-{\partial^2 _{\x_0,\x_1}R}(u,Du)Du=\\{\partial^2 _{\x_0,\x_1}R}(u,Du)[D^-_\e u-D u]+  D^-_\e u\big[
{\partial^2 _{\x_0,\x_1}R}(\zeta_\e(x),D^+_\e u)-{\partial^2 _{\x_0,\x_1}R}(u,Du)
\big].
\end{multline*}
The smoothness of $u$ and $D_\e u$ and the Lipschitz continuity of $\partial^2_{\x_0,\x_1}R$ yield\[\big|
{\partial^2 _{\x_0,\x_1}R}(\zeta_\e(x),D^+_\e u)-{\partial^2 _{\x_0,\x_1}R}(u,Du)
\big|\le C|D^+_\e u -D u|)\le \omega_\e,\]
so also this term goes to $0$ by the regularity of $u$.

Similarly,
\begin{multline*}
{\partial^2 _{\x_1^2}R}(u,\eta_\e(x))\D_\e u-{\partial^2 _{\x_1^2}R}(u,Du)\D u\\={\partial^2 _{\x_1^2}R}(u,Du)[\D_\e u-\D u]+[{\partial^2 _{\x_1^2}R}(u,\eta_\e(x))-{\partial^2 _{\x_1^2}R}(u,Du)]\D_\e u
\end{multline*}
The first part goes to $0$ by the regularity of $u$. By the boundedness of $u$, $Du$, $D^\pm_\e u$ and $\D_\e u$, we can use the Lipschitz continuity of ${\partial^2 _{\x_1^2}R}$ to get the bound
\[\big|
{\partial^2 _{\x_1^2}R}(u(x),\eta_\e(x),)-{\partial^2 _{\x_1^2}R}(u(x),Du(x))
\big|\le K|\eta_{\e}(x)- Du(x)|.\]
Since $\eta_{\e}(x)- Du(x)= D^+_\e u(x)-Du(x) +\e\mu_{\e,x} \D_\e u(x)$,
\[|\eta_{\e}(x)- Du(x)|\le |D^+_\e u(x)- Du(x)|+ \e\mu_{\e,x} |\D_\e u(x)|,
\]
and hence
\[
| {\partial^2 _{\x_1^2}R}(u(x),\eta_\e(x))-{\partial^2 _{\x_1^2}R}(u(x),Du(x))||\D_\e u(x)|\le ( |D^+_\e u(x)- Du(x)|+ \e\mu_{\e,x} |\D_\e u(x)|)|\D_\e u(x)|
\]
But 
\[
|\D_\e u(x)|\le|(\D_\e-\D )u(x)|+|\D u(x)|.
\]
By the propagation of the initial regularity, $\|\D u(\cdot,t)\|_\infty $ is bounded for any $t\in (0,T)$. Therefore
 \[|{\partial^2 _{\x_1^2}R}(u(x),\eta_\e(x))\D_\e u(x)-{\partial^2 _{\x_1^2}R}(u(x),Du(x))||\D_\e u(x)|\le\omega_\e. \] 
As for the  term $T^2_2$, we use again the chain rule:
\begin{multline*}T^2_2={D^-_\e}[{\partial_{\x_1} R}(u_\e,{D^+_\e }u_\e)]-{D^-_\e}[{\partial_{\x_1} R}(u,{D^+_\e }u)]=\\
{\partial^2 _{\x_0,\x_1}R}(\zeta_\e(x),D^+_\e u_\e)D^-_\e u_\e-{\partial^2 _{\x_0,\x_1}R}(\tilde\zeta_\e(x),D^+_\e u)D^-_\e u+
{\partial^2 _{\x_1^2}R}(u_\e,\eta_\e(x))\D_\e u_\e-{\partial^2 _{\x_1^2}R}(u,\tilde\eta_\e(x))\D_\e u,
\end{multline*}
where
\[
\zeta_\e(x)=u_\e(x)+\e\lambda_{\e,x} D^-_\e u_\e(x),\quad \tilde\zeta_\e(x)=u(x)+\e\lambda_{\e,x} D^-_\e u(x)
\]
\[
\eta_\e(x)= D^+_\e u_\e(x)+\e\mu_{\e,x}\D_\e u_\e(x),\quad \tilde\eta_\e(x)= D^+_\e u(x)+\e\mu_{\e,x}\D_\e u(x).
\]  
We use  the energy bound and Lemma \ref{sobo} to get the boundedness of $\zeta_\e$ and $D^\pm_\e u_\e$ and thus the Lipschitz continuity of $\partial^2_{\x_0,\x_1}R$, so that
\begin{multline*}
\big|{\partial^2 _{\x_0,\x_1}R}(\zeta_\e,D^+_\e u_\e)D^-_\e u_\e-{\partial^2 _{\x_0,\x_1}R}(\tilde\zeta_\e,D^+_\e u)D_\e^-u\big|\\
\le |{\partial^2 _{\x_0,\x_1}R}(\tilde\zeta_\e,D^+_\e u)|\/\/|D^-_\e u_\e-D^-_\e u|+K |D^-_\e u_\e|\/\/|\zeta_\e(x)-\tilde\zeta_\e(x)|.
\end{multline*}
But 
\[
|\zeta_\e(x)-\tilde\zeta_\e(x)|=|u_\e(x)-u(x)|+\e(|D^-_\e u_\e|+|D^-_\e u|),
\]
so that 
\[\int_{\mathcal{C}}|dx |\dot u-\dot u_\e||{D^+_\e}[{\partial_{\x_1} R}(u_\e,{D^+_\e }u_\e)]-{D_\e}[{\partial_{\x_1} R}(u,{D^+_\e }u)]|\le
 CW_\e +\frac 1 2 \e^2(\mathcal{E}[u(0)] +\mathcal{E}_\e[u_\e(0)])
\]

The term ${\partial^2 _{\x_1^2}R}(u_\e,\eta_\e(x))\Delta_\e u_\e-{\partial^2 _{\x_1^2}R}(u,\tilde\eta_\e(x))\D_\e u$ is more delicate because,  in order to use Lipschitz continuity, we need to bound  $\eta_\e(x)$ and hence the supremum of  $\D_\e u_\e$. But  Lemma \ref{sobo} and energy conservation are not enough when $n=2$. However we need really to bound $\e \D_\e u_\e$ and we can take advantage of this extra $\e$.
 Indeed, by the energy conservation and the positivity assumptions on $R$ and $A$, 
\[\e A_{n,n} (D_\e^nu_\e(x))^2\le C,\]
and hence
\be\label{bou} |D^n_\e u_\e|\le \frac {C}{\sqrt {\e}}, \ee
implying that $|\eta_{\e,x}|\le C$ (if $n>2$ we get a better estimate). Thus we have
\begin{eqnarray*}
&&\big|{\partial^2 _{1}R}(u_\e,\eta_\e )\D_\e u_\e-{\partial^2 _{1}R}(u,\tilde\eta_\e)\D_\e u\big|\le
 |{\partial^2 _{\x_1^2}R}(u,\tilde\eta_\e(x))|\/\/|\D_\e u-\D_\e u_\e|\\&&
 +|{\partial^2 _{\x_1^2}R}(u,\tilde\eta_\e(x))
 -{\partial^2 _{\x_1^2}R}(u,\eta_\e(x))|\/\/|\D_\e u_\e|+|{\partial^2 _{\x_1^2}R}(u_\e,\eta_\e(x))|-{\partial^2 _{\x_1^2}R}(u,\eta_\e(x))|\/\/|\D_\e u_\e|.
\end{eqnarray*}But, by (\ref{bou}), 
\[
|\eta_\e(x)-\tilde\eta_\e(x)|=|D^-_\e (u_\e(x)-u(x))|+\e(|\D_\e u_\e|+|\D_\e  u|)\le \omega_\e.\]
Therefore, by the Cauchy-Schwartz inequality and conservation of energy, we obtain that
\[\int_{\mathcal{C}} dx |\dot u-\dot u_\e| \big|{\partial^2 _{\x_1^2}R}(u_\e,\eta_\e )\D_\e u_\e-{\partial^2 _{\x_1^2}R}(u,\tilde\eta_\e)\D_\e  u\big|\le CW +\omega_\e.\]
Collecting all the terms, we conclude that there is are constant $C>0$  such that 
\[ \frac d{dt} W_\e\le C W_\e + \omega_\e,\]
and hence, by Gronwall lemma
\[|W_\e(t)|\le |W_\e(0)|\text{\rm e}^{Ct}+t\text{\rm e}^{Ct}\omega_\e \quad \text{ as } \e\to 0.\]

We summarize the results in the following 
\begin{theorem}
Suppose that $u_0$ and $v_0$ satisfy \eqref{ass02}.  Let $u(t)$   be the solution to (\ref {conteq}) and $u_\e(t)$ be the step function defined by (\ref{step}) with $u_i(t)$, $i=1,\dots,N$, solutions to (\ref{conteq1}) with initial data $u_i(0)=u_0(i\e)$ and $\dot u_i(0)=v_0(i\e)$. 
Then  for any $t>0$, 
\[\|u(t)-u_\e(t)\|\to 0, \quad \text{ as }\e \to 0,\]
in the norm defined in (\ref{norma}).
\end{theorem}

\end{document}